\documentclass[conference, 10pt]{IEEEtran}

\ifCLASSINFOpdf
\else
\fi

\usepackage{amssymb}
\usepackage{amsmath}
\usepackage{txfonts}
\usepackage{graphicx}
\usepackage{multirow}
\usepackage{setspace}
\usepackage{hyperref}
\usepackage{stfloats}

\newtheorem{lemma}{Lemma}
\newtheorem{thm}{Theorem}

\begin{document}
%
\title{Quickest Change Point Detection and Identification Across a Generic Sensor Array}
\author{\IEEEauthorblockN{Di Li}
\IEEEauthorblockA{Dept. of ECE\\
Texas A\&M University\\
College Station, TX 77843, USA} \and \IEEEauthorblockN{Lifeng Lai}
\IEEEauthorblockA{Dept. of ECE\\
Worcester Polytechnic Institute\\
Worcester, MA 01605, USA} \and \IEEEauthorblockN{Shuguang Cui}
\IEEEauthorblockA{Dept. of ECE\\
Texas A\&M University\\
College Station, TX 77843, USA}}

\maketitle

\begin{abstract}
In this paper, we consider the problem of quickest change point
detection and identification over a linear array of $N$ sensors,
where the change pattern could first reach any of these sensors, and
then propagate to the other sensors. Our goal is not only to detect
the presence of such a change as quickly as possible, but also to
identify which sensor that the change pattern first reaches. We
jointly design two decision rules: a stopping rule, which determines
when we should stop sampling and claim a change occurred, and a
terminal decision rule, which decides which sensor that the change
pattern reaches first, with the objective to strike a balance among
the detection delay, the false alarm probability, and the false
identification probability. We show that this problem can be
converted to a Markov optimal stopping time problem, from which some
technical tools could be borrowed. Furthermore, to avoid the high
implementation complexity issue of the optimal rules, we develop a
scheme with a much simpler structure and certain performance
guarantee.
\end{abstract}

\IEEEpeerreviewmaketitle

\section{Introduction}
The standard quickest change detection problem is set to detect some
unknown time point at which certain signal probability distribution
changes over a sequence of observations. Recently, with the
development of wireless sensor networks, multiple sensors can be
deployed to execute the quickest change detection, and the sensors
can send quantized or unquantized observations or certain local
decisions to a control center, who then makes a final
decision~\cite{Veeravalli:Decentralizedquickest,Tartakovsky:Asymptoticallydistributed}.
Most of the existing work is based on an assumption that the
statistical properties of observations at all sensors change
simultaneously. However, in certain scenarios, this assumption may
not hold well. For instance, when multiple sensors are used to
detect the occurrence of the chemical leakage, the sensors that are
closer to the leakage source usually observe the change earlier than
those far away from the source. In such cases, two interesting
problems arise: one is to detect the change as soon as possible; the
other is to identify which sensor is the closest to the source, such
that we could have a first-order inference over the leakage source
location.

As far as we know, currently there are few work studying the case of
change occurring non-simultaneously. In the related work, the
authors in \cite{Hadjiliadis:oneshot} proposed a scheme that each
sensor makes a local decision with the computing burden at the local
sensors, where they did not consider the identification problem. In
\cite{vvv:MarkovArray}, the authors modeled the change propagation
process as a Markov process to derive the optimal stopping rule and
assumed that the change pattern always first reaches a predetermined
sensor, such that the identification problem is ignored. In
\cite{Lai:identification}, the identification problem for the
special case of two sensors was studied, where the sufficient
statistic is proven as a Markov process and a joint optimal stopping
rule and terminal decision rule are proposed.

In this paper, we study the joint change point detection and
identification problem over a linear array of $N$ sensors, where the
change first occurs near an unknown sensor, then propagates to
sensors further away. We assume that all sensors send their
observations to a control center. With the sequential observation
signals, the control center first operates a stopping rule to decide
when to alarm that the change has occurred; then the control center
deploys a terminal decision rule to determine which sensor that the
change pattern reaches first. In our setup, three performance
metrics are of interest: i) detection delay, which is the time
interval between the moment that the change occurs and the moment
that an alarm is raised; ii) false alarm probability, which is the
probability that an alarm is raised before the actual change occurs;
and iii) false identification probability, which is the probability
that the control center does not correctly identify the sensor that
the change pattern first reaches. We apply the Markov optimal
stopping time theory to design the optimal decision rules to
minimize a weighted sum of the above three metrics. Furthermore, we
derive a scheme with a much simpler structure and certain
performance guarantee.

The rest of this paper is organized as follows. In
Section~\ref{sec:model}, we introduce the system model. In
Section~\ref{sec:solution}, we derive the optimal decision rules. In
Section~\ref{sec:approx}, we propose a scheme approximate to the
optimal decision rules with a much lower complexity. In
Section~\ref{sec:num}, we present some numerical results, with
conclusions in Section~\ref{sec:con}.

\section{System Model}\label{sec:model}

We consider a scenario with $N$ sensors constructing a linear array
to monitor the environment, as shown in Fig. \ref{array}. At an
unknown time point, a change occurs at an unknown location and
propagates, where we use change point time $\Gamma_i$ to denote the
time that the change pattern reaches sensor $i$. We further use $S$
to denote the index of the sensor that the change pattern first
reaches. We focus on the Bayesian setup and use $p_i$ to denote the
prior probability of $\{S=i\}$, simply with $p_i=1/N$. Conditioned
on the event that the change pattern first reaches sensor $i$,
$\Gamma_i$ is assumed to bear a geometric distribution
\cite{vvv:MarkovArray,Lai:identification} with parameter $\rho$,
$0<\rho\leq 1$ i.e.,
\begin{equation}
P[\Gamma_i=k|S=i]=\rho(1-\rho)^k,k\geq0,
\end{equation}
where $k$ denotes the discretized time and takes integer values.

We consider the practical factors in the environment, such as the
wind or the blockers, which will affect the propagation speed of the
change. For instance, see in Fig.~\ref{array}, if the direction of
the wind is from the left to the right side in the monitored
scenario, then the propagation of the air pollution will be much
faster at the right side of sensor $S=i$ than that of the left side.
And at the same side, the propagation follows the deterministic
order shown as
\begin{equation}
\left\{ \begin{array}{l}
 i \to i - 1 \to i - 2 \to i - 3...... \\
 i \to i + 1 \to i + 2 \to i + 3...... \\
 \end{array} \right.
\end{equation}

We further assume that after the change patten reaches the first
sensor $i$, for the right side of sensor $i$, it will propagate from
one sensor to another sensor following the geometric propagation
models as
\begin{equation}
P[\Gamma_{j+1}=k_1+k_2|\Gamma_j=k_1,S=i]=\rho_1(1-\rho_1)^{k_{2}},j>i,k_2\geq0,
\end{equation}
while for the left side of sensor $i$, the propagation follows
\begin{equation}
P[\Gamma_{j-1}=k_1+k_2|\Gamma_j=k_1,S=i]=\rho_2(1-\rho_2)^{k_{2}},j<i,k_2\geq0,
\end{equation}
where $\rho_1$ and $\rho_2$ are used to model possibly different
propagation speed along each direction, e.g., $\rho_1 > \rho_2$
means the propagation speed is higher at the right side that that of
the left side.

Taking above assumption, for $S=i$, we define all possible
$i\times(N-i+1)+1$ events at time $k$ as follows:
\[\begin{array}{l}
T_{0,k} = \{ {\Gamma _i} > k,{\Gamma _{i - 1}} > k,{\Gamma _{i + 1}} > k,{\Gamma _{i - 2}} > k,{\Gamma _{i + 2}} > k,...\}  \\
 T_{1,k} = \{ {\Gamma _i} \le k,{\Gamma _{i - 1}} > k,{\Gamma _{i + 1}} > k,{\Gamma _{i - 2}} > k,{\Gamma _{i + 2}} > k,...\}  \\
 T_{2,k} = \{ {\Gamma _i} \le k,{\Gamma _{i - 1}} > k,{\Gamma _{i + 1}} \le k,{\Gamma _{i - 2}} > k,{\Gamma _{i + 2}} > k,...\}  \\
 T_{3,k} = \{ {\Gamma _i} \le k,{\Gamma _{i - 1}} > k,{\Gamma _{i + 1}} \le k,{\Gamma _{i - 2}} > k,{\Gamma _{i + 2}} \le k,...\}  \\
  \vdots  \\
 T_{N - i + 1,k} = \{ {\Gamma _i} \le k,{\Gamma _{i - 1}} > k,{\Gamma _{i + 1}} \le k,...,{\Gamma _N} \le k\}  \\
 T_{N - i + 2,k} = \{ {\Gamma _i} \le k,{\Gamma _{i - 1}} \le k,{\Gamma _{i + 1}} > k,{\Gamma _{i - 2}} > k,{\Gamma _{i + 2}} > k,...\}  \\
 T_{N - i + 3,k} = \{ {\Gamma _i} \le k,{\Gamma _{i - 1}} \le k,{\Gamma _{i + 1}} \le k,{\Gamma _{i - 2}} > k,{\Gamma _{i + 2}} > k,...\}  \\
 T_{N - i + 4,k} = \{ {\Gamma _i} \le k,{\Gamma _{i - 1}} \le k,{\Gamma _{i + 1}} \le k,{\Gamma _{i - 2}} > k,{\Gamma _{i + 2}} \le k,...\}  \\
  \vdots  \\
 T_{2(N - i + 1),k} = \{ {\Gamma _i} \le k,{\Gamma _{i - 1}} \le k,{\Gamma _{i + 1}} \le k,...,{\Gamma _N} \le k\}  \\
  \vdots  \\
 T_{(i-1)(N - i + 1)+1,k} = \{ {\Gamma _i} \le k,{\Gamma _{i - 1}} \le k,{\Gamma _{i + 1}} > k,...,{\Gamma _1} \le k\}  \\
  \vdots  \\
 T_{i(N - i + 1),k} = \{ {\Gamma _i} \le k,{\Gamma _{i - 1}} \le k,{\Gamma _{i + 1}} \le k,...,{\Gamma _1} \le k,{\Gamma _N} \le k\} ,
 \end{array}\]
where $T_{1,k} \sim T_{N - i + 1,k}$ denote the events that after
the change pattern first reaching sensor $i$, it propagates across
the sensors sequentially at the right side of sensor $i$. The number
of events equals to the number of sensors at the right side plus 1,
i.e., $N - i + 1$. $T_{N - i + 2,k} \sim T_{2(N - i + 1),k}$ denote
the events that after the change pattern reaching sensor $i$ and
$i-1$, it propagates across the sensors sequentially at the right
side of sensor $i$, and the number of events is also $N - i + 1$,
which is the same for the case that after the change pattern
reaching sensor $i$, $i-1$, $i-2$ and so on. Since there are $i-1$
sensors at the left side of sensor $i$ and the event of no change
pattern reaches any sensor is $T_{0,k}$, the total number of
possible events is $i\times(N-i+1)+1$.

At each time $k$, we assume that the observations ${{\bf{z}}_k} =
[{{{z}}_{1,k}},...,{{{z}}_{N,k}}]$ from all sensors are available at
a control center. For each sensor $i$, conditioned on $\Gamma_i$,
${{z}}_{i,k}$ is Identically and Independently Distributed (IID)
according to $f_0$ before $\Gamma_i$, and IID according to $f_1$
after $\Gamma_i$, i.e.,
\begin{equation}\label{f0f1}
    z_{i,k} \sim \begin{cases}&f_0,~~~~  k<\Gamma_i,\\
&f_1,~~~~ k\geq\Gamma_i.
\end{cases}
\end{equation}

The observation sequence $\{\mathbf{z}_k;k=1,2,...\}$ generates a
filtration $\{\mathcal{F}_k;k=1,2,...\}$ with
\begin{equation}
    {{\cal F}_k} = \sigma({{\bf{z}}_1},...,{{\bf{z}}_k},\{\Gamma=0\}),k = 1,2,...
\end{equation}
where $\sigma ({{\bf{z}}_1},...,{{\bf{z}}_k},\{\Gamma  = 0\} )$
denotes the smallest $\sigma$-field in which
$({{\bf{z}}_1},...,{{\bf{z}}_k},\{\Gamma  = 0\} )$ is measurable and
$\Gamma=\min\{\Gamma_1,...,\Gamma_N\}$.

We use $P$ to denote the probability measure that specifies the
prior distribution of $\{S=i\}$, the distribution of change point
time, and the distribution of $\{\mathbf{z}_k;k=1,2,...\}$. We also
use $\mathbb{E}$ to denote the expectation under the probability
measure $P$. Specifically, we use $P_i$ to denote the probability
measure when $S=i$.

With above setups, the control center needs to detect the earliest
change point time $\Gamma=\min\{\Gamma_1,...,\Gamma_N\}$ as soon as
it occurs. A stopping time $\tau$ will be decided for when to stop
sampling and alarm that a change has occurred, where a false alarm
may happen if $\tau < \Gamma$. We target to minimize the averaged
detection delay $\mathbb{E}\{(\tau-\Gamma)^{+}\}$ with keeping the
false alarm probability $P[\tau < \Gamma]$ small. In addition, we
also require the control center to identify which sensor the change
pattern reaches first. We adopt $\delta_k$ to denote the
$\mathcal{F}_k$-measurable terminal decision rule used by the
control center to make the identification, and $\hat{S}$ to denote
the index of the sensor identified, i.e., $
\hat{S}=\delta_k(\mathcal{F}_k)$. A false identification occurs if
$\hat{S}\neq S$, such that we also want to keep $P[\hat{S}\neq S]$
small. We use $\boldsymbol{\delta}=\{\delta_1,\delta_2,...\}$ to
denote the sequence of terminal decision rules. Summarizing above,
our goal is to design a stopping time $\tau$ and a terminal decision
rule $\boldsymbol{\delta}$ that minimize the aggregated risk
function defined as
\begin{equation}\label{optimizationpro}
R\triangleq
P[\tau<\Gamma]+c_1\mathbb{E}\{(\tau-\Gamma)^{+}\}+c_2P[\hat{B}_1\neq
B_1],
\end{equation}
where $c_1$ and $c_2$ are appropriate constants that balance the
three costs.
\begin{figure}
\begin{center}
\includegraphics[width=0.5\textwidth]{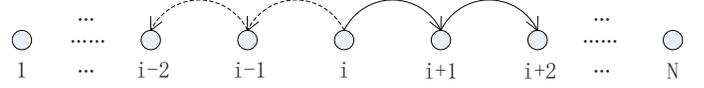}
\caption{A sensor array with $S=i$} \label{array}
\end{center}
\end{figure}

\section{Optimal Rules}\label{sec:solution}

In this section, we present the optimal stopping and terminal
decision rules. To proceed, we define the following posterior
probabilities at time $k$:
\begin{align}
&{\pi _{j,k\left| i \right.}} = P[T_{j,k}\left| {{{\cal F}_k},S = i}
\right.],~i =
1,...,N,~j = 0,...,i(N - i + 1), \label{def_pi}\\
 &p_k^i = P[{S} = i\left| {{{\cal F}_k}} \right.],~i = 1,...,N.\label{def_p}
 \end{align}
We also define the following matrices and vectors:
\begin{align}
 &{{\boldsymbol{\pi }}_k} = {[{\pi _{j,k\left| i \right.}}]_{M \times N}}, \\
 &{\bf{p}}_k = [p_k^1,p_k^2,...,p_k^N]_{1\times N},
 \end{align}
where $M=(N - \lfloor\frac{N+1}{2}\rfloor + 1) \times
\lfloor\frac{N+1}{2}\rfloor + 1$ denotes the maximum number of
events defined in Section II for $i\in\{1,...,N\}$. Corresponding to
$M$, $i=\lfloor\frac{N+1}{2}\rfloor$. For each $i$ with the number
of events less than $M$, the extra elements in ${\boldsymbol{\pi
}}_k$ are set as 0.

With the posterior probabilities defined above, we denote
$\mathbf{q}_k=\{{\boldsymbol{\pi }}_k,{\bf{p}}_k\}$.

We first have the following theorem regarding the optimal terminal
decision rule $\boldsymbol{\delta}$.
\begin{thm}
For any stopping time $\tau$, the optimal terminal decision rule is
\begin{equation}\label{decisonrule}
\hat{S}=\arg\max\limits_{1 \leq i \leq N}
\left\{p_\tau^{1},...,p_\tau^{i},...,p_\tau^{N}\right\},
\end{equation}
and we have
\begin{eqnarray}\label{decisionrule1}
\inf \limits_{\boldsymbol{\delta}}P[\hat{S}\neq
S]=\mathbb{E}\left\{1-\max\left\{p_\tau^{1},...,p_\tau^{N}\right\}\right\}.
\end{eqnarray}
\end{thm}
The proof follows from Proposition 4.1 of \cite{Poor:quickest}.
Theorem 1 implies that the optimal terminal decision rule is simply
to choose the sensor that has the largest posterior probability. A
similar situation also arises in the multiple hypothesis testing
problem considered in \cite{vvv:multihypothesis}.

Using above optimal terminal decision rule, we can further express
the optimization objective in \eqref{optimizationpro} as a function
of the posterior probabilities defined in (\ref{def_pi}) and
(\ref{def_p}), as shown below.
\begin{lemma}~\label{lem:cost} For any stopping time $\tau$, \eqref{optimizationpro} can be written as
\begin{equation}\label{formulation}
\begin{split}
R=\mathbb{E} & \left\{\sum\limits_{i = 1}^N {{{\pi _{0,\tau \left| i
\right.}}}} p_\tau ^i+c_2\left(1-\max\left\{p_\tau^{1},...,p_\tau^{N}\right\} \right)\right. \\
&~~+\left.{c_1}\sum\limits_{k = 0}^{\tau - 1} {\left(1 -
\sum\limits_{i = 1}^N {{\pi _{0,k\left| i \right.}}}
p_k^i\right)}\right\}.
\end{split}
\end{equation}
\end{lemma}

\begin{proof}
Based on the Bayesian's rule, we have
\begin{align}\label{eq1_lm1}
P[\Gamma  > k\left| {{{\cal F}_k}} \right.] = P[T_{0,k}\left|
{{{\cal F}_k}} \right.] &=
\sum\limits_{i = 1}^N {P[T_{0,k}\left| {{{\cal F}_k}} \right.,B = i]} p_k^i \nonumber\\
&= \sum\limits_{i = 1}^N {{\pi _{0,k\left| i \right.}}} p_k^i.
\end{align}

Further, according to Proposition 5.1 in \cite{Poor:quickest},
\begin{align}
&P[\tau<\Gamma]+c_1\mathbb{E}\{(\tau-\Gamma)^{+}\}\nonumber\\
&=\mathbb{E}\left\{P[\tau<\Gamma|\mathcal{F}_\tau]+c_1\sum
\limits_{k=0}\limits^{\tau-1}P[\Gamma\leq k|\mathcal{F}_k]\right\}\nonumber\\
&=\mathbb{E} \left\{ \sum\limits_{i = 1}^N {{\pi _{0,\tau \left| i
\right.}}} p_\tau ^i+ {c_1}\sum\limits_{k = 0}^{\tau - 1} \left(1 -
\sum\limits_{i = 1}^N {{\pi _{0,k\left| i \right.}}} p_k^i
\right)\right\}.
\end{align}
By combining (\ref{decisionrule1}), we complete the proof.
\end{proof}

Furthermore, we have the following lemma regarding $\mathbf{q}_k$.
\begin{lemma}
There is a time-invariant function $g$ such that
$\mathbf{q}_k=g(\mathbf{q}_{k-1},\mathbf{z}_k)$.
\end{lemma}

\begin{proof}
We have
\begin{align}\label{pi_jki}
&\pi_{j,k|i}=P[T_{j,k}|{\cal F}_{k-1},{{\bf{z}}_k},S=i]\nonumber\\
&= \frac{{f({{\bf{z}}_k}\left| {{{\cal F}_{k - 1}},T_{j,k}}
\right.,{S} = i)P[T_{j,k}\left| {{{\cal F}_{k - 1}},S=i}
\right.]}}{{\sum\limits_{j = 0}^{i(N -i+ 1)} {f({{\bf{z}}_k}\left|
{{{\cal F}_{k - 1}},T_{j,k}} \right.,{S} = i)P[T_{j,k}\left| {{{\cal
F}_{k - 1}},S=i} \right.]} }},
\end{align}
in which
\begin{align}\label{direcltywrite}
&f({{\bf{z}}_k}\left| {{{\cal F}_{k - 1}},T_{0,k}} \right.,S=i) ={f_0}({z_{i,k}}){f_0}({z_{i - 1,k}}){f_0}({z_{i + 1,k}})......\nonumber \\
&f({{\bf{z}}_k}\left| {{{\cal F}_{k - 1}},T_{1,k}} \right.,S=i) = {f_1}({z_{i,k}}){f_0}({z_{i - 1,k}}){f_0}({z_{i + 1,k}})...... \nonumber\\
&{\rm{ }} \vdots \nonumber \\
&f({{\bf{z}}_k}\left| {{{\cal F}_{k - 1}},T_{i(N-i+1),k}} \right.,S=i) = {f_1}({z_{i,k}}){f_1}({z_{i - 1,k}}){f_1}({z_{i + 1,k}})...... \nonumber\\
\end{align}

\begin{figure}
\begin{center}
\includegraphics[width=0.5\textwidth]{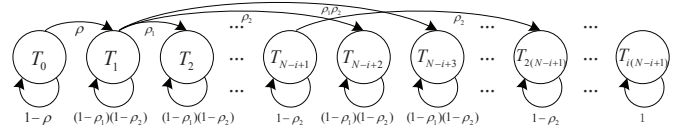}
\caption{Transition probabilities} \label{trans_prob}
\end{center}
\end{figure}

Since the geometric distribution based model owns the memoryless
property, the transition probabilities of the events are shown in
Fig.~\ref{trans_prob}. According to these transition probabilities,
we have
\begin{align}\label{transition}
&P[T_{0,k} | {{{\cal F}_{k - 1}},S = i}]=(1-\rho)\pi_{0,{k-1}|i}, \nonumber\\
&P[T_{1,k}| {{{\cal F}_{k - 1}},S = i }]= (1 - {\rho
_1})(1 - {\rho _2}){\pi_{1,k - 1\left| {i} \right.}} + \rho {\pi _{0,k - 1\left| {i} \right.}}\nonumber\\
&{\rm{ }} \vdots\nonumber\\
&P[T_{N-i+1,k}| {{{\cal F}_{k - 1}},S = i }]= (1 - {\rho _2}){\pi
_{N-i+1,k - 1\left| {i} \right.}} + \rho_1 {\pi _{N-i,k - 1\left| {i} \right.}}\nonumber\\
&P[T_{N-i+2,k}| {{{\cal F}_{k - 1}},S = i }]= (1 - {\rho _1})(1 -
{\rho _2}){\pi_{N-i+2,k - 1\left| {i} \right.}} + \rho_2 {\pi _{1,k - 1\left| {i} \right.}}\nonumber\\
&P[T_{N-i+3,k}| {{{\cal F}_{k - 1}},S = i }]= (1 - {\rho _1})(1 -
{\rho _2}){\pi_{N-i+3,k - 1\left| {i} \right.}} \nonumber\\
&~~~~~~~~~~~~~~~~~~~~~~~~~~~~+ \rho_1 {\pi _{N-i+2,k - 1\left| {i} \right.}}+\rho_2 {\pi _{2,k - 1\left| {i} \right.}}+\rho_1 \rho_2{\pi _{1,k - 1\left| {i} \right.}} \nonumber\\
&{\rm{ }} \vdots
\end{align}
Therefore, $\boldsymbol{\pi}_k$ can be computed by
$\mathbf{q}_{k-1}$ and $\mathbf{z}_k$.

For another element ${\bf{p}}_k$ in $\mathbf{q}_k$, we have
\begin{equation}\label{p_ki}
\begin{split} p_k^i& = P[{S} = i\left| {{{\cal F}_{k - 1}},{{\bf{z}}_k}} \right.] \\
&=\frac{{f({{\bf{z}}_k}\left| {{{\cal F}_{k - 1}},{S} = i}
\right.)P[{S} = i\left| {{{\cal F}_{k - 1}}}
\right.]}}{{\sum\limits_{n = 0}^{M-1} {f({{\bf{z}}_k}\left| {{{\cal
F}_{k - 1}},{S} = n} \right.)P[{S} = n\left| {{{\cal F}_{k - 1}}}
\right.]} }},
\end{split}
\end{equation}
where
\begin{align}\label{f_ki}
 &f({{\bf{z}}_k}\left| {{{\cal F}_{k - 1}},{S} = n} \right.) \nonumber\\
 &= \sum\limits_{j = 0}^{M-1} f({{\bf{z}}_k}\left| {{{\cal F}_{k - 1}},T_{j,k}}
 \right.,S=n)P[T_{j,k}| {{{\cal F}_{k - 1}},S = n }],
\end{align}
which can be calculated by using (\ref{direcltywrite}) and
(\ref{transition}). Hence, ${\bf{p}}_k$ can also be computed by
$\mathbf{q}_{k-1}$ and $\mathbf{z}_k$.
\end{proof}
This lemma implies that the posterior probabilities $\mathbf{q}_k$
can be recursively computed from $\mathbf{q}_{k-1}$ and
$\mathbf{z}_k$. Combined with Lemma~\ref{lem:cost}, we know that
$\mathbf{q}_k$ is a sufficient statistic for the problem of
minimizing (\ref{formulation}). Thus, the problem at the hand is a
Markov stopping time problem.

Therefore, we could borrow results from the optimal stopping time
theory to design the optimal decision rules for our problem. We
first consider a finite time horizon case, in which one has to make
a decision before a deadline $T$, i.e., $\tau\leq T$. It is easy to
check that the cost-to-go functions are
\begin{align}\label{cost-to-go}
J_T^{T}(\mathbf{q}_T)=\sum\limits_{i = 1}^N& {{\pi _{0,T\left| i
\right.}}} p_T^i+c_2\left(
1-\max\left\{p_T^{1},...,p_T^{N}\right\}\right), \texttt{and} \nonumber\\
J_k^{T}(\mathbf{q}_k)=\min&\left\{{\sum\limits_{i = 1}^N {{\pi
_{0,k\left| i \right.}}} p_k^i}+c_2\left(
1-\max\left\{p_k^{1},...,p_k^{N}\right\}\right), \right. \nonumber\\
&~~\left. c_1 \left(1 - \sum\limits_{i = 1}^N {{\pi _{0,k\left| i
\right.}}} p_k^i\right)+A_k^{T}(\mathbf{q}_k) \right\},
\end{align}
where
\begin{eqnarray}
A_k^{T}(\mathbf{q}_k)=\mathbb{E}\left\{J_{k+1}^T(\mathbf{q}_{k+1})|\mathcal{F}_k\right\}=\int
J_{k+1}^T(g(\mathbf{q}_k,\mathbf{z}))f(\mathbf{z}|\mathcal{F}_k)d\mathbf{z}.\nonumber
\end{eqnarray}

Applying the optimal stopping time theory \cite{Poor:quickest}, we
have the following theorem for the optimal decision rules.
\begin{thm}\label{thm2}
The optimal stopping time is obtained as
\begin{align}\label{stoprule}
\tau_{opt}=\inf&\left\{k: {\sum\limits_{i = 1}^N {{\pi _{0,k\left| i
\right.}}}
p_k^i}+c_2\left( 1-\max\left\{p_k^{1},...,p_k^{N}\right\}\right)\right.\nonumber\\
&~~~~~~~~\left.\leq  c_1 \left(1 - \sum\limits_{i = 1}^N {{\pi
_{0,k\left| i \right.}}} p_k^i\right)+A_k^{T}(\mathbf{q}_k)\right\},
\end{align}
with the optimal terminal decision rule is given in
(\ref{decisonrule}).
\end{thm}

In the infinite time horizon case when $T\rightarrow\infty$, we have
$J_k^\infty(\mathbf{q})$ defined as
\begin{equation}
J_k^\infty(\mathbf{q})=\lim_{T \rightarrow \infty}
J_k^T(\mathbf{q}),
\end{equation}
since we have $J_k^T(\mathbf{q})>0$, $J_k^{T+1}(\mathbf{q})\leq
J_k^T(\mathbf{q})$, and the fact that all strategies allowed with
deadline $T$ are also allowed with deadline $T+1$. Since the
observations are memoryless and conditionally IID,
$J_k^\infty(\mathbf{q})$ is the same for all $k$; we then use
$J(\mathbf{q})$ to denote $J_k^\infty(\mathbf{q})$. Thus,
$A(\mathbf{q})$ is derived as
\begin{equation}
    \begin{split}
    A(\mathbf{q})=\lim_{T\rightarrow\infty}A_k^T(\mathbf{q})&=\int \lim_{T\rightarrow\infty}
    J_{k+1}^T(g(\mathbf{q},\mathbf{z}))f(\mathbf{z}|\mathbf{q})d\mathbf{z}\\
     &=\int J(g(\mathbf{q},\mathbf{z}))f(\mathbf{z}|\mathbf{q})d\mathbf{z},
    \end{split}
\end{equation}
in which the interchange of $\lim$ and $\int$ is allowed due to the
dominated convergence theorem.

Therefore, when the deadline is infinite, the optimal stopping rule
becomes
\begin{align}\label{stoprule_infi}
\tau_{opt}=\inf&\left\{k: {\sum\limits_{i = 1}^N {{\pi _{0,k\left| i
\right.}}}
p_k^i}+c_2\left( 1-\max\left\{p_k^{1},...,p_k^{N}\right\}\right)\right.\nonumber\\
&~~~~~~~~\left.\leq  c_1 \left(1 - \sum\limits_{i = 1}^N {{\pi
_{0,k\left| i \right.}}} p_k^i\right)+A(\mathbf{q}_k)\right\},
\end{align}
with the optimal terminal decision rule is given in
(\ref{decisonrule}).
\section{Approximation to The Optimal Stopping Rule}\label{sec:approx}
When $N$ is large, the optimal stopping rule does not have a simple
structure, which makes the implementation highly costly. In this
section, we propose a much simpler rule which approximates to the
optimal stopping rule.

\begin{lemma}\label{lemma3}
The sequence $\left\{\min\limits_{1\leq i \leq
N}\{1-p_{k}^i\},\mathcal{F}_k;k\geq0\right\}$ is a supermartingale,
i.e.,
\begin{equation}\label{p2}
\mathbb{E}\left\{\min\limits_{1\leq i \leq
N}\{1-p_k^{i}\}|\mathcal{F}_{k-1}\right\}\leq \min\limits_{1\leq i
\leq N}\{1-p_{k-1}^i\},
\end{equation}
\end{lemma}
The proof follows from page 477 of \cite{Zacks:statistic}, by using
Fatou's lemma.

We can use Lemma~\ref{lemma3} to derive the following approximation
of the optimal stopping rule.

\begin{thm}
In the asymptotic case of the rare change occurring with
$\rho\rightarrow0$, one approximation of the optimal stopping rule
has the following simple structure
\begin{equation}\label{appro_stoppingrule}
    \tau_{app}=\inf\left\{k:\sum\limits_{j = 1}^{M-1} V_{k,j}\geq \frac{1}{c_1} \right\}
\end{equation}
where $V_{k,j}=\frac{\sum\limits_{i = 1}^N{{\pi _{j,k\left| i
\right.}}} p_{k}^i}{\rho\sum\limits_{i = 1}^N{{\pi _{0,k\left| i
\right.}}} p_{k}^i},~j=0,1,...,M-1$. And we use the optimal terminal
decision rule specified in (\ref{decisonrule}).
\end{thm}

\begin{proof}
The proof follows closely with the proof of Theorem 2 of
\cite{vvv:MarkovArray}. First, we have
\begin{align}\label{A(T-1)}
&A_{T-1}^{T}(\mathbf{q}_{T-1})=\mathbb{E}[J_T^{T}(\mathbf{q}_T)|\mathcal{F}_{T-1}]\nonumber\\
&=\int J_T^{T}[g(\mathbf{q}_{T-1},\mathbf{z}_T)]f(\mathbf{z}_T|\mathcal{F}_{T-1})d\mathbf{z}_T\nonumber\\
&=\int\left(\sum\limits_{i = 1}^N {{\pi _{0,T\left| i \right.}}}
p_T^i\right)f(\mathbf{z}_T|\mathcal{F}_{T-1})d\mathbf{z}_T\nonumber\\
&~~~~+\int\left[
c_2\left(1-p_T^{i_T}\right)\right]f(\mathbf{z}_T|\mathcal{F}_{T-1})d\mathbf{z}_T.
\end{align}
where $i_k=\arg\max\limits_{1\leq i \leq
N}\{p_{k}^{1},...,p_{k}^{N}\}$.

For the first part of (\ref{A(T-1)}), after interchanging the
integral and sum, by using (\ref{pi_jki}), (\ref{transition}),
(\ref{p_ki}), and (\ref{f_ki}), we have
\begin{align}\label{firstpart}
&\sum\limits_{i = 1}^N\int\left( {{\pi _{0,T\left| i \right.}}}
p_T^i \right)
f(\mathbf{z}_T|\mathcal{F}_{T-1})d\mathbf{z}_T\nonumber\\
&=\sum\limits_{i = 1}^N\int \frac{(1-\rho){\pi _{0,T-1\left| {i}
\right.}} {f({{\bf{z}}_T}\left| {{{\cal F}_{T - 1}},} \right.{S} =
i)}}{\sum\limits_{j = 0}^{M} f({{\bf{z}}_T}\left| {{{\cal F}_{T -
1}},T_{j,T}}
 \right.,S=i)P[T_{j,T}| {{{\cal F}_{T - 1}},S = i }]} \nonumber\\
&~~~~~\cdot\frac{{f({{\bf{z}}_T}\left| {{{\cal F}_{T - 1}},{S} = i}
\right.)p_{T-1}^{i}}}{{\sum\limits_{n = 1}^N {f({{\bf{z}}_T}\left|
{{{\cal F}_{T - 1}},{S} = n} \right.)P[{S} = n\left| {{{\cal F}_{T -
1}}} \right.]} }}
f(\mathbf{z}_T|\mathcal{F}_{T-1})d\mathbf{z}_T\nonumber\\
&=(1-\rho)\sum\limits_{i = 1}^N{{\pi _{0,T-1\left| i \right.}}}
p_{T-1}^i.
\end{align}

For the second part of (\ref{A(T-1)}), according to
Lemma~\ref{lemma3},
\begin{equation}\label{p2}
\int\left[
c_2\left(1-p_T^{i_T}\right)\right]f(\mathbf{z}_T|\mathcal{F}_{T-1})d\mathbf{z}_T
\leq c_2(1-p_{T-1}^{i_{T-1}}).
\end{equation}

Plugging the above two results (\ref{firstpart}) and (\ref{p2}) into
(\ref{A(T-1)}), we have
\begin{equation}\label{inequ}
A_{T-1}^{T}(\mathbf{q}_{T-1}) \leq (1-\rho)\sum\limits_{i =
1}^N{{\pi _{0,T-1\left| i \right.}}}
p_{T-1}^i+c_2(1-p_{T-1}^{i_{T-1}}).
\end{equation}

In the sequel, we assume that $A_{T-1}^{T}(\mathbf{q}_{T-1})$ equals
to the right side of (\ref{inequ}). 

According to (\ref{cost-to-go}), we have if $0\leq \sum\limits_{i =
1}^N{{\pi _{0,T-1\left| i \right.}}} p_{T-1}^i \leq
\frac{c_1}{c_1+\rho}$,
\[J_{T-1}^{T}(\mathbf{q}_{T-1})=\sum\limits_{i = 1}^N{{\pi _{0,T-1\left| i \right.}}} p_{T-1}^i+c_2 (1-p_{T-1}^{i_{T-1}}).\]
If $\frac{c_1}{c_1+\rho}\leq \sum\limits_{i = 1}^N{{\pi
_{0,T-1\left| i \right.}}} p_{T-1}^i \leq 1$,
\[J_{T-1}^{T}(\mathbf{q}_{T-1})=c_1+(1-\rho-c_1)\sum\limits_{i = 1}^N{{\pi _{0,T-1\left| i \right.}}}
p_{T-1}^i+c_2(1-p_{T-1}^{i_{T-1}}).\] We define the following
transformation as
\begin{align}\label{V_kl}
&V_{k,l}=\frac{\sum\limits_{i = 1}^N{{\pi _{l,k\left| i \right.}}}
p_{k}^i}{\rho\sum\limits_{i = 1}^N{{\pi _{0,k\left| i \right.}}}
p_{k}^i},~l=0,1,...,M-1.\nonumber\\&\Rightarrow\sum\limits_{j =
0}^{M-1}V_{k,j}=\frac{1}{\sum\limits_{i = 1}^N{{\pi _{0,k\left| i
\right.}}} p_{k}^i},V_{k,0}=\frac{1}{\rho}.
\end{align}
Then
\begin{equation}
{\sum\limits_{i = 1}^N{{\pi _{l,k| i }}} p_{k}^i}=\rho
V_{k,l}{\sum\limits_{i = 1}^N{{\pi_{0,k|i}}}
p_{k}^i}=\frac{V_{k,l}}{\sum\limits_{j =
1}^{M}V_{k,j}},~l=0,1,...,M-1.
\end{equation}
Further we have
\begin{align}
{\sum\limits_{i = 1}^N{{\pi _{l,k\left| i \right.}}}p_{k}^i}
=\frac{{\rho} V_{k,l}}{1+{\rho}{\sum\limits_{j = 1}^{M-1}V_{k,j}}},
~l=0,1,...,M-1
\end{align}
and
\begin{equation}
{\sum\limits_{i = 1}^N{{\pi _{0,k\left| i \right.}}}
p_{k}^i}=\frac{1}{1+{\rho}{\sum\limits_{j = 1}^{M-1}V_{k,j}}}.
\end{equation}
Then, $J_{T-1}^{T}$ can be rewritten as
\begin{align}\label{JTminus2}
&J_{T-1}^{T}(\mathbf{q}_{T-1})  = \nonumber\\
&\left\{ \begin{array}{l} \frac{1}{1+{\rho}{\sum\limits_{j =
1}^{M-1}V_{T-1,j}}}+c_2 (1-p_{T-1}^{i_{T-1}}),~\sum\limits_{j = 1}^{M-1} {V_{T-1,j}} \geq \frac{1}{c_1} \\
\frac{1-\rho+\rho c_1 {\sum\limits_{j = 1}^{M-1}V_{T-1,j}}
}{1+{\rho}{\sum\limits_{j = 1}^{M-1}V_{T-1,j}}}
+c_2(1-p_{T-1}^{i_{T-1}}),~\sum\limits_{j = 1}^{M-1} {V_{T-1,j}}
\leq \frac{1}{c_1}
 \end{array} \right..
\end{align}
We define $\mathbf{\Phi}_k$ and $\mathbf{\Psi}_k$ as
\begin{equation}
\mathbf{\Phi}_k\triangleq \frac{1}{1+{\rho}{\sum\limits_{j =
1}^{M-1}V_{k,j}}}+c_2 (1-p_{k}^{i_{k}})-J_k^T(\mathbf{q}_k),~0\leq k
\leq T,
\end{equation}
\begin{equation}
\mathbf{\Psi}_k\triangleq
A_k^T(\mathbf{q}_k)-\frac{1-\rho}{1+{\rho}{\sum\limits_{j =
1}^{M-1}V_{k,j}}}-c_2 (1-p_k^{i_k}),~0\leq k \leq T.
\end{equation}
Then straightly we see that $\mathbf{\Phi}_T=0$,
$\mathbf{\Psi}_{T-1}=0$, and
\begin{equation}
\mathbf{\Phi}_{T-1}=\left[\rho \frac{1-c_1 {\sum\limits_{j =
1}^{M-1}V_{T-1,j}}}{1+{\rho}{\sum\limits_{j = 1}^{M-1}V_{T-1,j}}}
\right]\mathbb{I}\left(\left\{\sum\limits_{j = 1}^{M-1} {V_{T-1,j}}
\leq \frac{1}{c_1}\right\}\right).
\end{equation}

For the next steps, we follow the proof of Theorem 2 of
\cite{vvv:MarkovArray}, which is skipped here. And additionally we
use Lemma~\ref{lemma3}.
Finally, it can be derived that
\begin{equation}
\lim\limits_{\rho\rightarrow0}
\frac{\mathbf{\Psi}_{T-k}}{\rho}\leq0,1\leq k \leq T.
\end{equation}

And the test structure reduces to stopping when
\begin{equation}
\sum\limits_{j = 1}^{M-1} {V_{k,j}} \geq \frac{1}{c_1}
\frac{1-\frac{\mathbf{\Psi}_{k}}{\rho
}}{1+\frac{\mathbf{\Psi}_{k}}{c_1 }}\geq \frac{1}{c_1}
\end{equation}

Therefore, we have the structure of the stopping rule as stated in
Theorem 3.
\end{proof}

Regarding to Theorem 3, we have several notes as follows.

1) From Lemma~\ref{lemma3} and Theorem~\ref{thm2}, we see that
$\tau_{app}$ is a lower bound of the optimal stopping time, i.e.
$\tau_{app}\leq \tau_{opt}$, in the case of $\rho\rightarrow0$. The
supermartingale property shown in Lemma~\ref{lemma3} plays an
important role in deriving $\tau_{app}$. The tightness of this lower
bound is related to the relationship between $\max\limits_{1\leq i
\leq N}p_{k}^{i}$ and $\mathbb{E}\left\{\max\limits_{1\leq i \leq
N}p_{k+1}^{i}|\mathcal{F}_{k}\right\}$. The simulation results in
Section~\ref{sec:num} show that $\max\limits_{1\leq i \leq
N}p_{k}^{i}$ and $\mathbb{E}\left\{\max\limits_{1\leq i \leq
N}p_{k+1}^{i}|\mathcal{F}_{k}\right\}$ are quite close, which
indicates that $\tau_{app}$ would be close to $\tau_{opt}$.


2) From (\ref{eq1_lm1}) and (\ref{V_kl}), we have the testing
statistic $\sum\limits_{j = 1}^{M-1} {V_{k,j}}$ as
\begin{align}
    \sum\limits_{j = 1}^{M-1} {V_{k,j}}=\frac{\sum\limits_{j = 1}^{M-1}\sum\limits_{i = 1}^N{{\pi _{j,k\left| i \right.}}}
p_{k}^i}{\rho\sum\limits_{i = 1}^N{{\pi _{0,k\left| i
\right.}}}p_{k}^i}&=\frac{1-\sum\limits_{i = 1}^N{{\pi _{0,k\left| i
\right.}}}p_{k}^i}{\rho\sum\limits_{i = 1}^N{{\pi _{0,k\left| i
\right.}}}p_{k}^i}\nonumber\\
&=\frac{1-P[\Gamma  > k\left| {{{\cal F}_k}} \right.]}{\rho P[\Gamma
> k\left| {{{\cal F}_k}} \right.]}.
\end{align}
This structure conforms to the well-known Shiryaev's
procedure~\cite{Mous:state-of-art}, which is the optimal stopping
rule for single sensor with IID observations and Bayesian setting.

\section{Numerical Simulation}\label{sec:num}
\begin{figure}
\begin{center}
\includegraphics[width=0.9\linewidth]{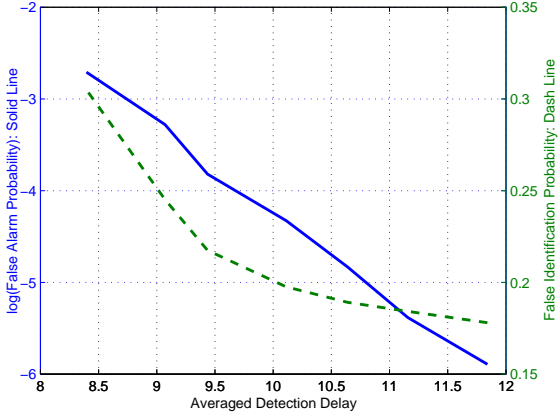}
\caption{False alarm and false identification probability vs.
averaged detection delay}\label{PFAvsADD}
\end{center}
\end{figure}

Given that it is hard to efficiently compute the solution structure
in (\ref{stoprule}), we compute the approximate optimal stopping
rule in (\ref{appro_stoppingrule}) and simulate its performance. We
assign 5 nodes constructing a linear sensor array and assume that
$f_0\sim \mathcal{N}(0,1)$ and $f_1\sim \mathcal{N}(1,1)$. The
change point time is generated according to the geometric
distribution with $\rho=0.01$, $\rho_1=0.1$ and $\rho_2=0.05$,
respectively. According to (\ref{eq1_lm1}), the false alarm
probability with $\tau_{app}$ is
\begin{equation}
    P[\tau_{app}\le \Gamma]=\mathbb{E}\left\{{\sum\limits_{i = 1}^N{{\pi _{N+1,\tau_{app}\left| i \right.}}}
p_{\tau_{app}}^i}\right\}\leq\frac{c_1}{c_1+\rho}=\alpha.
\end{equation}
Thus we have ${c_1}=\frac{\rho \alpha}{1-\alpha}$, where $\alpha$ is
the maximum allowance for the false alarm probability, which could
determine the required select $c_1$ value.

\begin{figure}
\begin{center}
\includegraphics[width=0.9\linewidth]{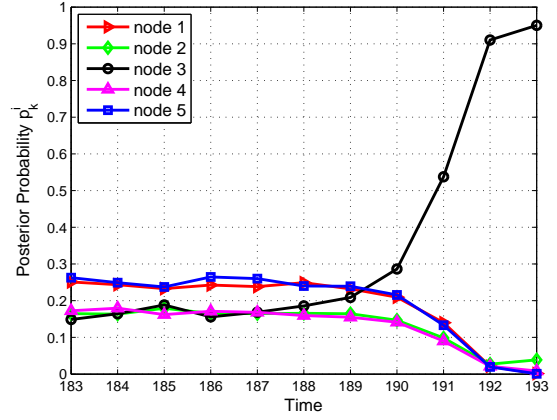}
\caption{Posterior probability of which node that the change pattern
first reaches}\label{p_figure}
\end{center}
\end{figure}

In Fig.~\ref{PFAvsADD}, we illustrate the relationships among the
false alarm probability, the false identification probability, and
the averaged detection delay. We see that as the averaged detection
delay increases, the false alarm probability decreases. When the averaged detection delay becomes large, 
the false identification probability does not decrease much and a probability floor appears, which is due to the fact that only the samples between the time when the change pattern reaches the first sensor and the time when it reaches the second sensor can be used to effectively distinguish the sensor that the change pattern first reaches. Since this part of the samples is limited, which will not increase with the detection delay, a false identification probability floor exists.  In Fig.~\ref{p_figure}, we draw the
posterior probability ${\bf{p}}_k^i$ over time, where we assume that
the change pattern first reaches node 3, and then propagates to node
4. We see that as time goes, $p_k^3$ gradually becomes larger than
the others, which indicates that node 3 should be identified. In
Fig.~\ref{bound_figure}, we show the relation between
$\max\{p_{k}^{1},...,p_{k}^{N}\}$ and
$\mathbb{E}\{\max\{p_{k+1}^{1},...,p_{k+1}^{N}\}|\mathcal{F}_{k}\}$
in (\ref{p2}). Since (\ref{p2}) is the key in deriving the our
simplified rule, the fact that these two curves are close suggests
that the performance of our low-complexity rule might be close to
that of the optimal stopping rule in (\ref{stoprule}) and
(\ref{stoprule_infi}).
\begin{figure}
\begin{center}
\includegraphics[width=0.9\linewidth]{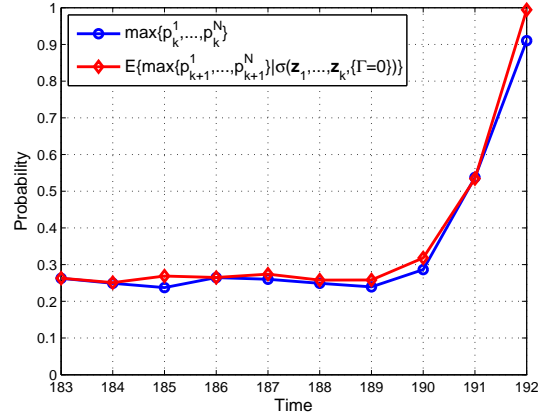}
\caption{$\max\limits_{1\leq i \leq N}p_{k}^{i}$~vs.~$
\mathbb{E}\{\max\limits_{1\leq i \leq
N}p_{k+1}^{i}|\mathcal{F}_{k}\}$}\label{bound_figure}
\end{center}
\end{figure}
\section{Conclusions}\label{sec:con}
We have studied the quickest change point detection problem and the
closest-node identification problem over a sensor array. We have
proposed an optimal decision scheme combing the stopping rule and
the identification rule to alarm the change happening and to
determine the sensor closest to the change source. Since the
structure the obtained optimal scheme is complex and impractical to
implement, we have further proposed a scheme with a much simpler
structure.

\end{document}